\documentclass[12pt]{iopart}

\usepackage[skip=0.6\baselineskip]{caption}
\newtheorem{theorem}{Theorem}[section]

\newtheorem{proposition}[theorem]{Proposition}

\newtheorem{definition}[theorem]{Definition}
\newenvironment{proof}[1][Proof]{\begin{trivlist}
\item[\hskip \labelsep {\bfseries #1}]}{\end{trivlist}}

\newcommand{\qed}{\nobreak \ifvmode \relax \else
	\ifdim\lastskip<1.5em \hskip- \lastskip
	\hskip1.5em plus0em minus0.5em \fi \nobreak
	\vrule height0.75em width0.5em depth0.25em\fi}

\usepackage{hyperref}
\hypersetup{backref=true,pagebackref=true,hyperindex=true,colorlinks=false,breaklinks=true,urlcolor= blue,bookmarks=true,bookmarksopen=false,pdftitle={book title},pdfauthor={book author}}
\usepackage{iopams}  
\begin{document}

\title[Some results on marginally trapped tubes]{Some results on cosmological and astrophysical horizons and trapped surfaces}

\author{A M Sherif$^1$, R Goswami$^2$ and S D Maharaj$^3$}

\address{Astrophysics and Cosmology Research Unit, School of Mathematics, Statistics and Computer Science, University of KwaZulu-Natal, Private Bag X54001, Durban 4000, South Africa}
\ead{\mailto{abbasmsherif25@gmail.com}, \mailto{Goswami@ukzn.ac.za}, \mailto{Maharaj@ukzn.ac.za}}
\vspace{10pt}
\begin{indented}
\item[]May 2019
\end{indented}

\begin{abstract}
We study the evolution of horizons of black holes in the \(1+1+2\) covariant setting and investigate various properties intrinsic to the geometry of the foliation surfaces of these horizons. This is done by interpreting formulations of various quantities in terms of the geometric and thermodynamic quantities. We establish a causal classification for horizons in different classes of spacetimes. We have also recovered results by Ben-Dov and Senovilla which put cut-offs on the equation of state parameter \(\sigma\), determining the spacelike, timelike and non-expanding horizons in the the Robertson-Walker class of spacetimes. We show that stability of marginally trapped surfaces (MTS) in the Robertson-Walker spacetimes is only achievable under the conditions of negative pressure, and also classify the spacelike future outer trapping horizons (SFOTH) in the Robertson-Walker spacetimes via bounds on the equation of state parameter \(\sigma\). For the Lemaitre-Tolman-Bondi (LTB) model, it is shown that a relationship between the energy density and the electric part of the Weyl curvature, \(\mathcal{E}\), gives the causal classification of the MTTs. It is further shown that only spacelike MTTs are foliated by stable MTS, and that this stability guarantees no shell crossing. We also provide an explicit proof of the third law of black hole thermodynamics for the LRS II class of spacetimes, and by extension, any spacetime whose outgoing and ingoing null geodesics are normal to the MTS. 
\end{abstract}

\vspace{2pc}
\maketitle

\section{Introduction}

Being one of the most engaged areas of research in recent years, the study of black holes has taken a central role in understanding Einstein's theory of general relativity. The notion of black holes is understood in context of what are known as closed trapped surfaces \cite{ge2,josh1,haw1}. The formation of a closed trapped surface in a spacetime \(M\) signals gravitational collapse and generally indicates geodesic incompleteness of \(M\). These closed trapped surfaces are determined by the sign of the null expansion scalars along the null vector fields normal to them \cite{ge2,josh1,haw1,hay1,ash1,ash2,sen1,ibb1}. We can make a choice of one of the null expansion vanishing with the other strictly negative to define what are called marginally trapped surfaces (MTS) (we usually take the null expansion along the outgoing null direction to vanish). These surfaces foliate a \(3\)-dimensional submanifold of \(M\). It has been shown in \cite{ib1}, under the assumption that the intrinsic geometry and the normal bundles connection on the closed trapped \(2\)-surfaces are fixed, that any such surface serves as a suitable boundary in a quasi-local action formulation of general relativity. As such these \(3\)-submanifolds are suitable black hole boundaries and are used to study the local dynamics and evolution of black holes \cite{hay1,ash1,ash2}.

Various properties of such submanifolds have been studied. The uniqueness of dynamical horizons - a spacelike \(3\)-submanifold \(H\), foliated by MTS with signature of the induced metric fixed all over \(H\) and their geometric properties, were investigated in \cite{ash3}. A notion of stability of the leaves of the foliation, the MTS, analogous to minimal surfaces in Riemannian geometry has also been established \cite{and1,and2,yau1}. 

The topological properties of these MTS have also been investigated by various authors \cite{ge2,gal1,rn1,gal2,friedman1,gal3,j1,as1} producing some very interesting results. For example, Stephen Hawking theorized \cite{ge2} that cross sections of the event horizon for asymptotically flat and stationary spacetimes, satisfying the dominant energy condition, are topological \(2\)-spheres. This is the well known Hawking black hole topology theorem. In 1987,  Newmann constrained the result by Hawking by showing that the cross sections have to satisfy certain stability conditions \cite{rn1}. 

Ellis \textit{et al.}, \cite{rit1} studied the evolution of the horizons, which the authors called MOTS, foliated by MTS, in locally rotationally symmetric class II spacetimes. They introduced conditions on the slope of the tangent to the MOTS curves which determine the nature of the MOTS. This led the authors to describe MOTS in a real astrophysical setting and found that an initial MOTS bifurcates into an outer and inner MOTS and that the inner MOTS was timelike while the outer MOTS was spacelike.

The interest in these horizons as local description of black holes gave rise to the formulation of the laws of black hole mechanics and even flux laws locally. Trapping horizons, introduced by Hayward was first used to formulate dynamical laws of black hole mechanics \cite{hay3}. Ashtekar and coauthors introduced the notion of an isolated horizon (IH), to characterize the equilibrium states of black holes. Truly dynamical black holes are characterized by dynamical horizon (DH) \cite{hay1,ash1,ash2,sen1}. It has been shown that flux laws can as well be formulated on DH \cite{ash1,ash2}.

The laws of black hole thermodynamics are formulated in context of the properties of quantities such as the surface gravity and angular momentum on the horizon \cite{ash7,ash8,shh1,by1}. These are formulated so that they are intrinsic to the horizon geometry. For example the notion of surface gravity, which is closely related to the temperature of a black hole, also plays an important role in the study of black hole dynamics \cite{hay3, ash7,ash8,ash9}.

In this paper we investigate the evolution of black hole horizons in context of the \(1+1+2\) covariant splitting variables. The formulation in this paper allows us to make some very important observations about black hole horizons. We determine the forms of the horizon function determining the sign of the intrinsic metric on the horizon for certain classes of spacetimes as well as for the general case. We use these results to study horizons and trapped surfaces in some well known spacetimes. We then consider some results concerning surface gravity on the horizon. 

Because of the many different notations to be used in the paper, we list in tabe \ref{table:tabe} below abbreviations to be used in this paper.

\begin{table}[ht]
\caption{List of abbreviations.}
\centering
\begin{tabular}{|c|c|c|c|}
\hline \hline
\small{Abbreviation} & Name & \small{Abbreviation} & Name \\
\hline
\small{\textbf{TS}} & \small{Trapped Surface} & \small{\textbf{NEH}} & \small{Nonexpanding Horizon}\\
\small{\textbf{MTS}} & \small{Marginally Trapped Surface} & \small{\textbf{IH}} & \small{Isolated Horizon}\\
\small{\textbf{MTT}} & \small{Marginally Trapped Tube} & \small{\textbf{DH}} & \small{Dynamical Horizon}\\
\small{\textbf{LRS}} & \small{Locally Rotationally Symmetric} & \small{\textbf{TLM}} & \small{Timelike Membrane}\\
\small{\textbf{NNF}} & \small{Null Normal Foliation} & \small{\textbf{OS}} & \small{Oppenheimer-Snyder}\\
\small{\textbf{WEC}} & \small{Weak Energy Condition} & \small{\textbf{RW}} & \small{Robertson-Walker}\\
\small{\textbf{DEC}} & \small{Dominant Energy Condition} & \small{\textbf{LTB}} & \small{Lemaitre-Tolman-Bondi}\\
\small{\textbf{NEC}} & \small{Null Energy Condition} & \small{\textbf{(S)FOTH}} & \small{(Spacelike) Future Outer}\\
\ & \  & \ & \small{Trapping Horizon}\\
\hline
\end{tabular}
\label{table:tabe}
\end{table}

This paper is structured as follow: In section \ref{covf1} we give an overview of the \(1+1+2\) covariant splitting of spacetime following \cite{rit1,cc1}. Section \ref{ppt} looks at the nature of horizons in various classes of spacetimes and use the results to study properties of horizons and marginally trapped surfaces in the Robertson-Walker and the Lemaitre-Tolman-Bondi spacetimes. In section \ref{0n} we calculate the surface gravity for various classes of spacetimes as well as the general case, and in section \ref{err} we conclude with discussion of the results of the paper.

\section{\label{covf1}Preliminaries}
In this section we introduce the \(1+1+2\) semi-tetrad covariant formalism \cite{rit1,cc1,rit2} and some definitions. 

\subsection{Some notes on the \(1+1+2\) covariant formalism}\label{}

Given a spacetime \(M\), to covariantly describe \(M\) we first make a choice of a unit tangent vector field \(u^a\), to a defined timelike congruence \cite{gfre1}. Given any \(4\)-vector \(U^a\) in the spacetime, the projection tensor \(h_a^{\ b}\equiv g_a^{\ b}+u_au^b\), projects \(U^a\) onto the \(3\)-space as

\begin{eqnarray*}
U^a&=&Uu^a + U^{\langle a \rangle },
\end{eqnarray*}
where \(U\) is the scalar along \(u^a\) and \(U^{\langle a \rangle }\) is the projected \(3\)-vector. This naturally gives rise to two derivatives:
\begin{itemize}
\item The \textit{covariant time derivative} (or simply the dot derivative)  along the observers' congruence. For any tensor \(S^{a..b}_{\ \ \ \ c..d}\), \(\dot{S}^{a..b}_{\ \ \ \ c..d}\equiv u^e\nabla_eS^{a..b}_{\ \ \ \ c..d}\).

\item Fully orthogonally \textit{projected covariant derivative} \(D\) with the tensor \(h_{ab}\), with the total projection on all the free indices. For any tensor \(S^{a..b}_{\ \ \ \ c..d}\), \(D_eS^{a..b}_{\ \ \ \ c..d}\equiv h^a_{\ f}h^p_{\ c}...h^b_{\ g}h^q_{\ d}h^r_{\ e}\nabla_rS^{f..g}_{\ \ \ \ p..q}\).
\end{itemize}
The kinematical and Weyl quantities associated with the \(1+3\) splitting is given by the set \(\lbrace{\rho,p,\Theta,H_{ab},E_{ab},\sigma_{ab},\pi_{ab},q_a,\omega_a\rbrace}\). We have \(\rho=T_{ab}u^au^b\) as the energy density, \(p=\frac{1}{3}h^{ab}T_{ab}\) is the isotropic pressure, \(\Theta=D_au^a\) is expansion, \(H_{ab}= *C_{acbd}u^cu^d\) is the gravito-magnetic tensor, \(E_{ab}= C_{acbd}u^cu^d\) is the gravito-electric tensor (where \(C_{acbd}\) is the Weyl tensor and \(*\) denotes the dual), \(\sigma_{ab}\) is the shear tensor, \(\pi_{ab}=\pi_{\langle ab \rangle}\) is the anisotropic stress tensor, \(q_a=q_{\langle a \rangle}=-h_a^{\ b}T_{bc}u^c\) is the \(3\)-vector defining the heat flux, and \(\omega_a\) is the rotation vector. The quantity \(T_{ab}\) is the energy-momentum tensor.

The \(3\)-space can further be split by the choice of a vector field - which we denote by \(e^a\) -  orthogonal to \(u^a\), and the kinematical and Weyl variables described above can be split into an irreducible set of covariant scalar variables. The vector field \(e^a\) allows for the introduction of a projector tensor given by \(N_a^{\ b}\equiv g_a^{\ b}+u_au^b-e_ae^b\), which projects vectors orthogonal to \(u^a\) and \(e^a\) onto a \(2\)-surface defined as the sheet (\(N_a^{\ a}=2\)). This further splitting introduces two new derivatives:
\begin{itemize}
\item The \textit{hat derivative} is the spatial derivative along the vector \(e^a\). For a \(3\)-tensor \(\psi_{a..b}^{\ \ \ \ c..d}\), \(\hat{\psi}_{a..b}^{\ \ \ \ c..d}\equiv e^fD_f\psi_{a..b}^{\ \ \ \ c..d}\).

\item The \textit{delta derivative} is the projected spatial derivative on the \(2\)-sheet by \(N_a^{\ b}\) and projected on all the free indices. For any \(3\)-tensor \(\psi_{a..b}^{\ \ \ \ c..d}\), \(\delta_e\psi_{a..b}^{\ \ \ \ c..d}\equiv N_a^{\ f}..N_b^{\ g}N_h^{\ c}..N_i^{\ d}N_e^{\ j}D_j\psi_{f..g}^{\ \ \ \ h..i}\).
\end{itemize} 

The resulting quantities fully describing a spacetime \(M\) are then given by the set \(\lbrace{\rho, p, A,\Theta,\phi, \Sigma, \mathcal{E}, \Pi, Q, \Omega, \xi, \mathcal{H}_{ab},\mathcal{E}_{ab},\Sigma_{ab},\Pi_{ab}, A^a, \mathcal{H}^a,\mathcal{E}^a,\Sigma^a,\Pi^a, q^a, \omega^a\rbrace}\) which are defined as follows

\begin{eqnarray}\label{maro1}
\fl \eqalign{&\sigma_{ab}=\Sigma\left(e_ae_b-\frac{1}{2}N_{ab}\right)+2\Sigma_{(a}e_{b)}+\Sigma_{ab}, \ q_a=Qe_a, \ \dot{u}_a=Ae_a+A_a,  \ \phi=\delta_ae^a,\\
\fl&\omega_a=\Omega e_a+\Omega_a,}
\end{eqnarray}

\begin{eqnarray}\label{maro17}
\fl E_{ab}=\mathcal{E}\left(e_ae_b-\frac{1}{2}N_{ab}\right)+2\mathcal{E}_{(a}e_{b)}+\mathcal{E}_{ab},\ \pi_{ab}=\Pi\left(e_ae_b-\frac{1}{2}N_{ab}\right)+2\Pi_{(a}e_{b)}+\Pi_{ab}.
\end{eqnarray}

In \eref{maro1} and \eref{maro17}, \(A_a\) is the acceleration vector, \(\phi\) is the sheet expansion, and \(\mathcal{E}\), \(\Pi\) and \(\Sigma\) are the scalars associated with the electric part of the Weyl tensor, the anisotropic stress and the shear tensor respectively. The evolution and propagation equations can be obtained from the Ricci identities of the vectors \(u^a\) and \(e^a\) as well as the doubly contracted Bianchi identities. Some of the evolution and propagation equations relevant to this paper are given below (see \cite{cc1}, which contains the complete set of equations):

\begin{itemize}
\item \textit{Evolution:}
\begin{eqnarray}\label{evoo1}
\fl\eqalign {\frac{2}{3}\dot{\Theta}-\dot{\Sigma}=&A\phi-2\left(\frac{1}{3}\Theta - \frac{1}{2}\Sigma\right)^2-\frac{1}{3}\left(\rho+3p-2\Lambda\right)+\mathcal{E}-\frac{1}{2}\Pi-\Sigma_a\Sigma^a + \Omega_a\Omega^a\\
&-\left(2a_a - A_a - \delta_a\right)A^a+2\Omega^2+\varepsilon_{ab}\alpha^a\Omega^b -2\alpha_a\Sigma^a-\Sigma_{ab}\Sigma^{ab},}
\end{eqnarray}

\begin{eqnarray}\label{paaaaa2}
\fl\eqalign{\dot{\phi}=&\left(\frac{2}{3}\Theta-\Sigma\right)\left(A-\frac{1}{2}\phi\right)+Q+2\xi\Omega+\delta_a\alpha^a-\zeta^{ab}\Sigma_{ab}+A^a\left(\alpha_a - a_a\right)\\
&+\left(a^a-A^a\right)\left(\Sigma_a-\varepsilon_{ab}\Omega^b\right).}
\end{eqnarray}
\item \textit{Propagation:}
\begin{eqnarray}\label{evoo2}
\fl\frac{2}{3}\hat{\Theta}-\hat{\Sigma}&=&\frac{3}{2}\phi \Sigma + Q + 2\xi\Omega + \delta_a\Sigma^a+\varepsilon_{ab}\delta^a\Omega^b - 2\Sigma_aa^a + 2\varepsilon_{ab}A^a\Omega^b-\Sigma_{ab}\zeta^{ab},
\end{eqnarray}

\begin{eqnarray}\label{paaaaa1}
\fl\eqalign{\hat{\phi}=&\left(\frac{1}{3}\Theta+\Sigma\right) \left(\frac{2}{3}\Theta-\Sigma\right)-\frac{1}{2}\phi^2 -\frac{2}{3}\left(\rho+\Lambda\right)-\mathcal{E}-\frac{1}{2}\Pi+2\xi^2 + \delta_aa^a \\
&- a_aa^a - \zeta_{ab}\zeta^{ab}+ 2\varepsilon_{ab}\alpha^a\Omega^b-\Sigma_a\Sigma^a + \Omega_a\Omega^a.}
\end{eqnarray}
\end{itemize}

The full covariant derivatives of the vectors \(u^a\) and \(e^a\) are given by \cite{cc1}

\begin{eqnarray}\label{001}
\fl\eqalign{\nabla_au_b=&-Au_ae_b + e_ae_b\left(\frac{1}{3}\Theta + \Sigma\right) + N_{ab}\left(\frac{1}{3}\Theta -\frac{1}{2}\Sigma\right)-u_aA_b+e_a\left(\Sigma_b+ \varepsilon_{bm}\Omega^m\right)\\
&+\Omega\varepsilon_{ab} +\left(\Sigma_a-\varepsilon_{am}\Omega^m\right)e_b+\Sigma_{ab},}
\end{eqnarray}

\begin{eqnarray}\label{00004}
\fl\eqalign{\nabla_ae_b=&-Au_au_b + \left(\frac{1}{3}\Theta + \Sigma\right)e_au_b +\frac{1}{2}\phi N_{ab}-u_a\alpha_b+\left(\Sigma_a-\varepsilon_{am}\Omega^m\right)u_b+e_aa_b \\
&+ \xi\varepsilon_{ab}+\zeta_{ab},}
\end{eqnarray}
where

\begin{eqnarray*}
\fl\eqalign{\alpha_a\equiv \dot{e}_a=N_{ab}\dot{e}^b,\ \varepsilon_{ab}\equiv \varepsilon_{abc}e^c=u^d\eta_{dabc}e^c, \ \zeta_{ab}\equiv \delta_{\lbrace{a}} e_{{b}\rbrace},\xi\equiv\frac{1}{2}\varepsilon^{ab}\delta_ae_b,\ a_a\equiv e^cD_ce_a=\hat{e}_a.}
\end{eqnarray*}
The quantities \(\Sigma,\Sigma_a,\Sigma_{ab}\) are related to the shear tensor and shear scalar via the relations

\begin{eqnarray}\label{006}
\eqalign{\sigma^2\equiv \frac{1}{2}\sigma_{ab}\sigma^{ab}=\frac{3}{4}\Sigma^2+\Sigma_a\Sigma^a + \frac{1}{2}\Sigma_{ab}\Sigma^{ab}.}
\end{eqnarray}
We also have the relation \(\hat{u}_a=\left(\frac{1}{3}\Theta+\Sigma\right)e_a+\Sigma_a+\varepsilon_{ab}\Omega^b\). The quantity \(\zeta_{ab}\) is the shear of \(e^a\) (distortion of the sheet), \(a^a\) is its acceleration, and \(\xi\) is the twisting of the sheet (rotation of \(e^a\)). 

\subsection{Some definitions}

Let us now define some notions useful to this paper. As was mentioned in the introduction, the authors in \cite{ib1} showed that a \(3\)-surface foliated by marginally trapped \(2\)-surfaces is a suitable boundary of a black hole under certain conditions. In the definitions and discussions that are to follow, \(k^a\) and \(l^a\) are respectively the outward and inward null normal vector fields to a leaf of such foliation, while \(\Theta_k\) and \(\Theta_l\) are the expansions of the congruences generated by \(k^a\) and \(l^a\) respectively.

\begin{definition}[\textbf{Trapped Surface}]\label{nggg}
A (future) trapped surface (TS) is a smooth, connected, closed, spacelike co-dimension \(2\) submanifold \(S\) of \(M\) such that the divergences, \(\Theta_k\) and \(\Theta_l\), of the congruences generated by the null normal vector fields \(k^a\) and \(l^a\) (\(k^a\) is the outgoing null normal vector field and \(l^a\) is the ingoing null normal vector field) respectively are everywhere negative on \(S\).  
\end{definition} 

\begin{definition}[\textbf{Marginally Trapped Surface}]\label{nggg0}
A marginally trapped surface (MTS) is a smooth, connected, closed, spacelike co-dimension \(2\) submanifold \(S\) of \(M\) such that \(\Theta_k\) is everywhere vanishing on \(S\) and \(\Theta_l\) is everywhere negative on \(S\). 
\end{definition}

\begin{definition}[\textbf{Marginally Trapped Tube}]\label{nggg1}
A marginally trapped tube (MTT) is a co-dimension 1 submanifold \(H\) of \(M\) which is foliated by MTS. 
\end{definition}
In the case that there is no restriction on the sign of \(\Theta_l\), then the co-dimension \(1\) surface is called a \textit{marginally outer trapped tube} (MOTT), and the co-dimension \(2\) surfaces foliating the MOTT are called \textit{marginally outer trapped surfaces} (MOTS). For more on the above definitions see the following references (\cite{hay1,ash1,ash2,sen1,sen3,sen4,sen5,ibb1}, as well as \cite{rit1}). In general, the signature of the induced metric on \(H\) will vary over \(H\). There are however cases where the signature is fixed all over \(H\) (\(C\) is a fixed real number everywhere on \(H\)). In such cases a spacelike MTT is called a dynamical horizon (DH), a timelike MTT is called a timelike membrane (TLM), and a null and non-expanding MTT is called an isolated horizon (IH). An MTT is a \textit{future outer trapping } if and only if \(\mathcal{L}_l\Theta_k\) is everywhere negative on \(H\) and a \textit{future inner trapping horizon} if and only if \(\mathcal{L}_l\Theta_k\) is everywhere positive on \(H\), where \(\mathcal{L}_l\) is the Lie derivative operator in the direction of the null normal vector field \(l^a\). Timelike implies inner trapped and spacelike implies outer trapped.

We now introduce certain classes of spacetimes to be considered throughout this work. An LRS spacetime \(M\) is a spacetime in which at each point \(p\in M\), there exists a continuous isotropy group generating a multiply transitive isometry group on \(M\) \cite{crb1,gbc1,stge1,vege1, ssgos1,ssgos2}. The general metric of LRS spacetimes is given by

\begin{eqnarray}\label{jan29191}
\fl \eqalign{ds^2=&-A^2dt^2 + B^2d\chi^2 + F^2 dy^2 + \left[\left(F\bar{D}\right)^2+ \left(Bh\right)^2 - \left(Ag\right)^2\right]dz^2\\ &+ \left(A^2gdt - B^2hd\chi\right)dz,}
\end{eqnarray}
where \(A^2,B^2,F^2\) are functions of \(t\) and \(\chi\), \(\bar{D}^2\) is a function of \(y\) and \(k\) (\(k\) fixes the geometry of the \(2\)-surfaces), and \(g,h\) are functions of \(y\). LRS II spacetimes is a subclass of LRS spacetimes with \(g=0=h\). 

The field equations and full covariant derivatives of \(u^a\) and \(e^a\) for LRS spacetimes are given by the vanishing of all tensor and vector quantities in \eref{evoo1} to \eref{00004}. For LRS II spacetimes, in addition to the vanishing of the tensor and vector quantities in \eref{evoo1} to \eref{00004}, we also have the vanishing of \(\Omega\) and \(\xi\).

LRS II class of spacetimes generalizes spherically symmetric solutions to Einstein field equations (EFEs). Examples of physically interesting spherically symmetric solutions that fall into the class of LRS II spacetimes include Schwarzschild, Friedman-Lemaitre-Robertson-Walker (FLRW), Lemaitre-Tolman-Bondi, Vaidya and the Oppeinheimer-Snyder dust solutions. LRS spacetimes, a generalization of LRS II spacetimes, include solutions with nonzero vorticity and nonzero spatial twist. Some of these solutions include the G\"{o}del's world model, the Kantowski-Sachs models and the Bianchi models, invariant under the \(G_3\) groups of types I, II, VIII and IX (see for example the reference \cite{gft2}).

For LRS II spacetimes, the outward and inward pointing null normal vectors to the MTS - given as co-dimension \(2\) smooth embeddings of \(M\) - are given by 

\begin{eqnarray}\label{oct20}
k^a=\frac{1}{\sqrt{2}}\left(u^a+e^a\right), \quad l^a=\frac{1}{\sqrt{2}}\left(u^a-e^a\right),
\end{eqnarray}
respectively (see \cite{rit1,rit2}). 

\begin{definition}\label{deeef1}
We will denote by NNF (for null normal foliation) the class of spacetimes for which there exists horizons foliated by MTS, and where null geodesics (both outgoing and ingoing) are normal to these MTS. Locally rotationally symmetric (LRS) spacetimes is a subclass of NNF. 
\end{definition}
The class of NNF spacetimes provides a certain generalization of LRS spacetimes in the sense that, a spacetime may possess the property that the preferred scalings of the null vector fields in the spacetime coincide with those of LRS spacetimes, but without all of the symmetries of LRS spacetimes. In the case of NNF spacetimes, there are nonscalar background quantities, i.e. under linear perturbations, vector and tensor quantities are not necessarily gauge invariant.

For LRS II spacetimes, the expansions \(\Theta_k\) and \(\Theta_l\) are given as 
\begin{eqnarray}\label{oct21}
\Theta_k=\frac{1}{\sqrt{2}}\left(\frac{2}{3}\Theta-\Sigma+\phi\right), \quad \Theta_l=\frac{1}{\sqrt{2}}\left(\frac{2}{3}\Theta-\Sigma-\phi\right),
\end{eqnarray}
respectively (see \cite{as1} for more details). 

\begin{definition}[Hayward]\label{def02}
A spacelike future outer trapping horizon (SFOTH) is a spacelike MTT \(H\) of spacetime, such that \(\mathcal{L}_l\Theta_k<0\). 
\end{definition}

\begin{definition}[Regularity]\label{def01}
A DH \(\mathcal{H}\) is regular if\\
(1). \(H\) is achronal (for \(p,q\in H\), \(\nexists\) a timelike curve \(\gamma:\left(a,b\right)\longrightarrow H\) such that \(\gamma\left(a\right)=p\) and \(\gamma\left(b\right)=q\) or \(\gamma\left(a\right)=q\) and \(\gamma\left(b\right)=p\)).\\
\ \\
(2). \(2\sigma^2+T_{ab}k^ak^b\neq 0\) on \(H\) \(\forall\) null vectors \(k^a\), where \(T_{ab}\) is the energy momentum tensor.
\end{definition}

We next look at evolution of MTT in LRS II and NNF classes using the prescription in \cite{ib3}, in context of the \(1+1+2\) splitting (also see \cite{ib1,ib2}).

\section{On Properties of MTTs}\label{ppt}

For NNF spacetimes (\(k^a\) and \(l^a\) can be written in the form \eref{oct20}), it turns out that the null expansion scalars take exactly the same form as \eref{oct21}. This result was obtained in \cite{as1} and implies topological equivalence of horizons types in these spacetimes to LRS II. In this section, we consider evolution of MTTs. Specifically, we study the evolution of MTT in the LRS II class of spacetimes, the NNF class, then a general \(4\)-dimensional spacetime. The approach in this work for determining the causal character of an MTT utilizes the formulation in \cite{ib3}. We compute a certain smooth function on the MTT determining its causal character, denoted \(C\), in terms of quantities from the \(1+1+2\) decomposition of spacetime admitting the MTT \cite{cc1,rit1}. This allows us to phrase the definition of the causal nature of an MTT in terms of constraints on the geometric and thermodynamic variables. We will see that the causal character of the MTT is determined by some energy condition on the MTT, which is expected as stated in \cite{ib3}.  The comparison of the three cases considered also provides us with a useful way to determine a certain equivalence under causal properties. Obtaining expressions in terms of these geometric and thermodynamic variables provides us with useful insights into the nature and properties of MTTs of certain well studied spacetimes. Some interesting statements can be made about MTTs and marginally trapped \(2\)-surfaces in some cosmological and astrophysical examples provided. Some existing results are also recovered. We emphasize that while our focus is on MTTs, the results in this section equally applies to MOTTs.

\subsection{Evolution of MTTs}

We now describe the procedure for determining the signature of the induced metric on an MTT \(H\), following \cite{ib3}, where the null expansions now are expressed in terms of quantities from the the \(1+1+2\) splitting. It is worth mentioning that in \cite{ib3}, the analyses were restricted to spherically symmetric spacetimes. Consequently, by adapting the calculation of \(C\) to the variables from the \(1+1+2\) splitting, we can extend to LRS II spacetimes. While it is true that in general this approach may not be suitable for extensive analysis of the evolution of MTTs, we will see that the form of \(C\) for more general spacetimes provides insight into causal relationships of MTTs in LRS II spacetimes and more general spacetimes.

Let us first briefly discuss the formalism in \cite{ib3}. Introduce a vector field \(\bar{X}^a\) which is normal to an MTT \(H\). A tangent vector field \(X^a\) to \(H\) (in the sense that \(X^a\bar{X}_a=0\)) is also introduced which is everywhere orthogonal to the foliation. The vector field \(X^a\) generates a foliation preserving flow (\(\mathcal{L}_Xv=f\left(v\right)\), for some function \(f\left(v\right)\)  (\(v\) labels the foliation)). Both \(X^a\) and \(\bar{X}^a\) are assumed to be future pointing in the sense that

\begin{eqnarray*}
X^al_a<0, \ \bar{X}^al_a<0.
\end{eqnarray*}
There is a further requirement that

\begin{eqnarray*}
\bar{X}^a\bar{X}_a&=&-X^aX_a.
\end{eqnarray*}
Of course if \(H\) is spacelike then \(X^aX_a>0\) and similarly \(\bar{X}^a\bar{X}_a<0\). The vector fields \(X^a\) and \(\bar{X}^a\) can be written as 

\begin{eqnarray}\label{01}
X^a=\alpha\left(k^a-Cl^a\right), \quad \bar{X}^a=\alpha\left(k^a+Cl^a\right),
\end{eqnarray}  
respectively, where \(C\) is some scalar field on \(H\) (\(C\in C^\infty\left(H\right)\)) and \(\alpha\in\mathbb{R}^+\). Without loss of generality, we set \(\alpha=1\), and from the definition of \(X^a\),

\begin{equation}\label{02}
\mathcal{L}_X\Theta_k=\mathcal{L}_{\left[k-Cl\right]}\Theta_k=0,
\end{equation}  
from which we write the explicit expression for the field \(C\):

\begin{equation}\label{03}
C=\frac{\mathcal{L}_k\Theta_k}{\mathcal{L}_l\Theta_k}.
\end{equation} 
The proportionality of \(X^aX_a\) and \(\bar{X}^a\bar{X}_a\) to \(C\) means that the sign of \(C\), at a point of \(H\), can be used to determine the causal character of the MTT at the point: If \(C<0\) the MTT is timelike, if \(C>0\) the MTT is spacelike and if \(C=0\) or \(C=\infty\) (\(\mathcal{L}_k\Theta_k\neq 0\) and \(\mathcal{L}_l\Theta_k=0\)) the MTT is null. The sign of \(C\) also determines whether the MTT is expanding (\(C>0\)), contracting (\(C<0\)) or unchanging in area (\(C=0\) or \(C=\infty\)). Thus the MTT is timelike if and only if it is contracting and spacelike if and only if it is expanding.

As an example, we can explicitly calculate \(C\) for the LRS II class of spacetimes in terms of the scalars of the \(1+1+2\) formalism: 

\begin{equation}\label{04}
\eqalign{\mathcal{L}_k\Theta_k&=k^a\nabla_a\Theta_k\\
&=\dot{\Theta}_k+\hat{\Theta}_k\\
&=-\left(\rho+p\right)-\Pi+2Q,}
\end{equation} 
and

\begin{equation}\label{05}
\eqalign{\mathcal{L}_l\Theta_k&=l^a\nabla_a\Theta_k\\
&=\dot{\Theta}_k-\hat{\Theta}_k\\
&=\frac{1}{3}\left(\rho+3p\right)+2\mathcal{E},}
\end{equation}
which gives

\begin{equation}\label{06}
C_{LRS\ II}=\frac{-\left(\rho+p\right)-\Pi+2Q}{\frac{1}{3}\left(\rho-3p\right)+2\mathcal{E}}.
\end{equation}

The quantities \(\dot{\Theta}_k\) and \(\hat{\Theta}_k\) are computed using the appropriate linear combination of \eref{evoo1} to \eref{paaaaa1}. We note that the Lie derivative is along the null normals and evaluated on the MTS (thus on the MTT). We can verify the validity of the expression in \eref{06} by checking some simple cases. For example, consider the case for the spherically symmetric vacuum spacetimes. In this case, \(\rho=p=\Pi=Q=0\) and so the numerator of \eref{06} is zero. Since \(\mathcal{E}\neq 0\), \(H\) is null, i.e. \(H\) is an NEH. Considering Oppenheimer-Snyder collapse, \(p=\Pi=Q=\mathcal{E}=0\), and \(C=-3\). Thus \(H\) is a TLM and contracting. As some authors have argued that TLMs cannot be associated with the surface of a black hole during its evolution since timelike curves can traverse the membrane in both directions, this is a case where a TLM is clearly associated with the evolution of the black hole. 

Since an MTT \(H\) is timelike if and only if it is inner trapped (collapsing), i.e. \(\mathcal{L}_l\Theta_k>0\), for the LRS II class of spacetimes we can rephrase this in terms of the matter and scalar variables as follows: 
\begin{proposition}\label{rem1}
\textit{An MTT in a spacetime in the LRS II class of spacetimes is timelike if and only if \(\rho>3p-6\mathcal{E}\)}. 
\end{proposition}
But as the MTT is timelike, \(C<0\), which implies \(\left(\rho+p\right)+\Pi-2Q>0\). Suppose we consider an ideal case like the perfect fluid, we have \(\Pi=Q=0\). If the weak energy condition (WEC) is satisfied then the condition for the MTT to be timelike is that \(\rho>3p-6\mathcal{E}\) which agrees with the result in \cite{rit1}. Similarly, a MTT in a spacetime in the LRS II class is spacelike if and only if \(\rho\leq 3p-6\mathcal{E}\).

Similarly, we can compute \(C\) for NNF spacetimes. We shall denote by \(\left(\ast\right)_{LRS\ II}\) the part of a quantity \(\ast\) restricted to those covariant scalars completely describing LRS II spacetimes (the covariant scalars completely describing LRS II spacetimes are given by the set \(\lbrace{\rho, p, A,\Theta,\phi, \Sigma, \mathcal{E}, \Pi, Q\rbrace}\)).

\begin{equation}\label{oct1}
\eqalign{\mathcal{L}_k\Theta_k&=k^a\nabla_a\Theta_k\\
&=\left(\dot{\Theta}_k+\hat{\Theta}_k\right)\vline_{LRS\ II}+\left(T_1 + T_2\right)+\varepsilon_{ab}\left(R^a+\bar{R}^a\right)\Omega^b,}
\end{equation} 
and

\begin{equation}\label{oct2}
\eqalign{\mathcal{L}_l\Theta_k&=l^a\nabla_a\Theta_k\\
&=\left(\dot{\Theta}_k-\hat{\Theta}_k\right)\vline_{LRS\ II}+\left(T_1 - T_2\right)+\varepsilon_{ab}\left(R^a-\bar{R}^a\right)\Omega^b,}
\end{equation}
where we set

\begin{eqnarray*}
T_1&=&3\Omega^2-\Sigma^2+2\Omega\xi+\delta_a\left(\mathcal{A}^a+\alpha^a\right)+\left(R_a-2a_a\right)\mathcal{A}^a+\left(a^a-\mathcal{A}^a\right)\Sigma_a\\
&&-2\alpha_a\Sigma^a-\Sigma_{ab}\left(\Sigma^{ab}+\zeta^{ab}\right),\\
T_2&=&2\xi^2+\Omega^2+2\xi\Omega-\left(\zeta_{ab}+\Sigma_{ab}\right)\zeta^{ab}-\left(2a^a+\Sigma^a\right)\Sigma_a+\delta_a\left(a^a+\Sigma^a\right),\\
R^a&=&\mathcal{A}^a+\alpha^a-a^a,\\
\overline{R}^a&=&\delta^a+2\left(\alpha^a+\mathcal{A}^a\right).
\end{eqnarray*}
This gives

\begin{eqnarray}\label{oct3}
C_{NNF}=\frac{\left(\dot{\Theta}_k+\hat{\Theta}_k\right)\vline_{LRS\ II}+\left(T_1 + T_2\right)+\varepsilon_{ab}\left(R^a+\bar{R}^a\right)\Omega^b}{\left(\dot{\Theta}_k-\hat{\Theta}_k\right)\vline_{LRS\ II}+\left(T_1 - T_2\right)+\varepsilon_{ab}\left(R^a-\bar{R}^a\right)\Omega^b}.
\end{eqnarray}
Suppose we have a case where \(R^a=\bar{R}^a\) and \(T_1 = T_2\). Then \eref{oct3} can be written as

\begin{eqnarray}\label{oct4}
C_{NNF}=C\vline_{LRS\ II}+\bar{C},
\end{eqnarray}
where the function \(\bar{C}\) is given by

\begin{eqnarray}\label{oct5}
\bar{C}=6\frac{T_1 +\varepsilon_{ab}R^a\Omega^b}{\rho-3p+6\mathcal{E}}.
\end{eqnarray}
While not of any physical relevance and no physical motivation for the particular case given here, from \eref{oct4} we can conclude that, in general, though MTT in NNF are topologically equivalent to LRS II, the MTT will evolve differently. We note that the conditions on \(R^a\) and \(\bar{R}^a\) are actually conditions on \(\varepsilon_{ab}R^a\Omega^b\) and \(\varepsilon_{ab}\bar{R}^a\Omega^b\). 

We now provide some useful definitions and theorems.

We state the following result, which is a result noted in \cite{ash3}, with proof making use of \eref{06}:
\begin{theorem}\label{th1}
Let \(M\) be a \(4\)-dimensional spacetime and let \(H\) be a DH in \(M\) satisfying the NEC. If \(H\) satisfies the second condition of regularity, then \(H\) is a (future) outer trapping horizon (FOTH).
\end{theorem}
\begin{proof}
Condition (2) in definition \ref{def01} implies \(\mathcal{L}_l\Theta_k\) is nowhere vanishing on \(H\) \cite{ash3}. If \(H\) satisfies the NEC, then by the Raychaudhuri equation, we know that \(\mathcal{L}_k\Theta_k\) is nonpositive on \(H\). Using \eref{03}, since \(H\) is spacelike, \(\mathcal{L}_k\Theta_k\) is strictly negative (\(\mathcal{L}_k\Theta_k=0\) implies \(H\) is null). This implies that \(\mathcal{L}_l\Theta_k\) has to be everywhere negative on \(H\) (for \(C\) to be positive). \qed
\end{proof}

Such DH \(H\) is the spacelike (future) outer trapping horizon (SFOTH) put forward by Hayward, and whose existence is an indicator of the presence of a black hole \cite{ash3}.

Using \eref{04} and \eref{05}, the signs of \(\mathcal{L}_k\Theta_k\) and \(\mathcal{L}_l\Theta_k\) on an SFOTH \(H\) imply the following result:

\begin{theorem}\label{th2}
A DH \(H\) in a spacetime \(M\) of the LRS II class of spacetimes satisfying the energy condition

\begin{eqnarray}\label{10}
2Q-\Pi<\rho+p<4p-6\mathcal{E},
\end{eqnarray}
\end{theorem}
is an SFOTH.

Theorem \ref{th2} thus gives us a relatively easy way of detecting the possible existence of a black hole in LRS II spacetimes. Similar results can be stated for NNF class. For example a \(DH\) in a spacetime in the NNF class satisfying the condition \(2Q-\Pi-p+\left(T_1+T_2\right)+\varepsilon_{ab}\left(R^a+\bar{R}^a\right)\Omega^b<\rho<3p-6\mathcal{E}-3\left(T_1-T_2\right)-3\varepsilon_{ab}\left(R^a-\bar{R}^a\right)\Omega^b\) is an SFOTH.

Combining the result on topological equivalence between horizon types in the LRS II and NNF classes, and the case considered in \eref{oct4}, we state the following result:

\begin{theorem}\label{th10}
Let \(H'\) be an MTT in the NNF class and let \(H'\) satisfy the case considered in \eref{oct4}. Suppose \(H\) is a null MTT in the LRS II class to which \(H'\) is topologically equivalent. Furthermore, suppose \(H'\) is regular and satisfies the NEC. Then \(H'\) is an SFOTH if \(T_1 +\varepsilon_{ab}R^a\Omega^b<0\) and a TLM if \(T_1 +\varepsilon_{ab}R^a\Omega^b>0\).
\end{theorem}
This result also follows from the fact that \(H'\) being regular and satisfying the NEC implies that \(\mathcal{L}_l\Theta_k\) is everywhere negative on \(H\). From theorem \ref{th10} it is clear that it is entirely possible, in principle, to have topological equivalence between MTTs satisfying the same energy conditions, but exhibiting different causal character. We will state this as follows:

\begin{proposition}\label{rem2}
\textit{In general, topologically equivalent MTTs will exhibit different causal character.}
\end{proposition}
This is known and is expected in general as the causal character relies on a much \textit{stronger} structure than the topology (see chapter 2 of \cite{bd1}).

For a general \(4\)-dimensional spacetime, the generators of the outgoing and ingoing null geodesics to the two surfaces contain a sheet component, i.e. there is a component of \(k^a,l^a\), denoted \(\tau^a,-\tau^a\) respectively, projected onto the \(2\)-sheet via \(N_{ab}\). The generators for the outgoing and ingoing null geodesics are written as

\begin{eqnarray}\label{nov151}
k^a=\frac{1}{\sqrt{2}}\left(u^a+e^a+\tau^a\right), \ l^a=\frac{1}{\sqrt{2}}\left(u^a-e^a-\tau^a\right),
\end{eqnarray}
respectively. There is a coefficient, \(\tau\) (the magnitude of \(k^a\) along \(e^a\)), of \(e^a\) for both null normals, but can be set to unity without loss of generality. We calculate the null expansions in the \(k^a\) direction and obtain

\begin{eqnarray}\label{nov152}
\Theta_k&=&\frac{1}{\sqrt{2}}\left(\frac{2}{3}\Theta-\Sigma+\phi+W\right),
\end{eqnarray}
where \(W=\nabla_a\tau^a=\left(\delta_a - a_a\right)\tau^a\) (this decomposition, obtained from the decomposition of the fully orthogonally projected covariant derivative, \(D\), of a vector orthogonal to \(u^a\) and \(e^a\), can be found in \cite{cc1}). It is clear that for a general \(4\)-dimensional spacetime, if the sheet component of the null normal vector fields to the MTS is divergence free, then these MTS are topologically LRS II (see \cite{as1}). We therefore have

\begin{eqnarray}\label{nov154}
\fl\mathcal{L}_k\Theta_k&=&\left(\dot{\Theta}_k+\hat{\Theta}_k\right)\vline_{LRS\ II}+\left(T_1 + T_2\right)+\varepsilon_{ab}\left(R^a+\bar{R}^a\right)\Omega^b+\dot{W}+\hat{W}+\tau^a\delta_a\Theta_k,\\
\fl\mathcal{L}_l\Theta_k&=&\left(\dot{\Theta}_k-\hat{\Theta}_k\right)\vline_{LRS\ II}+\left(T_1 - T_2\right)+\varepsilon_{ab}\left(R^a-\bar{R}^a\right)\Omega^b+\dot{W}-\hat{W}-\tau^a\delta_a\Theta_k,
\end{eqnarray}
so that

\begin{eqnarray}\label{nov155}
\fl C=\frac{\left(\dot{\Theta}_k+\hat{\Theta}_k\right)\vline_{LRS\ II}+\left(T_1 + T_2\right)+\varepsilon_{ab}\left(R^a+\bar{R}^a\right)\Omega^b+\dot{W}+\hat{W}+\tau^a\delta_a\Theta_k}{\left(\dot{\Theta}_k-\hat{\Theta}_k\right)\vline_{LRS\ II}+\left(T_1 - T_2\right)+\varepsilon_{ab}\left(R^a-\bar{R}^a\right)\Omega^b+\dot{W}-\hat{W}-\tau^a\delta_a\Theta_k},
\end{eqnarray}
where 

\begin{eqnarray}\label{nov153}
\hat{\Theta}_k=\frac{1}{\sqrt{2}}\left(\frac{2}{3}\hat{\Theta}-\hat{\Sigma}+\hat{\phi}+\hat{W}\right), \ \dot{\Theta}_k=\frac{1}{\sqrt{2}}\left(\frac{2}{3}\dot{\Theta}-\dot{\Sigma}+\dot{\phi}+\dot{W}\right).
\end{eqnarray}

\subsection{A causal classification}\label{000000001}

We have now considered evolution of MTTs for general \(4\)-dimensional spacetimes, NNF and LRS II spacetimes. The approach used here, where we have combined the formulation in \cite{ib3} coupled with the \(1+1+2\) decomposition of spacetime allow us to give a particular classification of MTTs causally, i.e. classifying causally equivalent MTTs.

Given \(C\), we will assume that the sign of denominator and numerator of \(C\) remains fixed by translation. This then leaves the sign of \(C\) and therefore the causal character of the MTT unchanged (notice that this ensures that causally equivalent MTTs satisfy the same energy conditions). It is also clear that if the support of certain functions is non-empty, \(C\) changes sign as well. This implies that the overall causal character of these MTTs will change at these support points. It is with these in mind we proceed with the following classification.  

\subsubsection{LRS II and NNF}\label{sshor1}
\ \\
\ \\
Suppose we want to look at the particular case in \eref{oct4}. For a null MTT \(H\) in the LRS II class of spacetimes, i.e. \(C\vline_{LRS\ II}=0\), we may determine the causal character of the corresponding MTT, \(H'\), in the NNF class, which is topologically equivalent to \(H\), by some restrictions on the function \(\bar{C}\). For such case as that of \eref{oct4}, the evolution of the MTT will be precisely LRS II if \(T_1=-\varepsilon_{ab}R^a\Omega^b\). In fact, in general, MTT in the NNF class satisfying the condition that \(\bar{R}^a=T_2=0\) will be both causally and diffeomorphically LRS II. 

\subsubsection{NNF and the general case}
\ \\
\ \\
An MTT in a general \(4\)-dimensional spacetime satisfying one of the conditions below is causally NNF:
\begin{itemize}\label{}
\item[a.] \(\hat{W}=-\tau^a\delta_a\Theta_k\).
\item[b.] \(\dot{W}=0\) and \(\hat{W}=-\tau^a\delta_a\Theta_k\).
\item[c.] \(W=0\) and \(\tau^a\delta_a\Theta_k=0\).
\end{itemize}
In the case of the last item c., the MTT is also diffeomorphically NNF. 

\subsubsection{LRS II and the general case}
\ \\
\ \\
An MTT in a general \(4\)-dimensional spacetime satisfying one of the conditions below is causally LRS II:
\begin{itemize}\label{}
\item[a.] \(\hat{W}=-\tau^a\delta_a\Theta_k\) and \(T_2=\bar{R}^a=0\).
\item[b.] \(\dot{W}=0\), \(\hat{W}=-\tau^a\delta_a\Theta_k\) and \(T_2=\bar{R}^a=0\).
\item[c.] \(W=0\), \(\tau^a\delta_a\Theta_k=0\) and \(T_2=\bar{R}^a=0\).
\item[d.] The case considered in \eref{oct4} and the additional condition that \(T_1=-\varepsilon_{ab}R^a\Omega^b\), with \(\hat{W}=-\tau^a\delta_a\Theta_k\), \(\dot{W}=0\) and \(\hat{W}=-\tau^a\delta_a\Theta_k\), or \(W=0\) and \(\tau^a\delta_a\Theta_k=0\).
\end{itemize}
In the cases of the items c. and d., the MTT is also diffeomorphically LRS II. 

This classification is by no means exhaustive as just as we considered in \eref{oct4} where we looked at a particular splitting of \(C\), similar conditions can be put on \(C\) to obtain further classifications.

Next we use the expression for \(C\) for the LRS II class, \eref{06}, and investigate the MTTs in some well known spacetimes. We obtain some known results like the bounds on the equation of state parameter \(\sigma\) in RW spacetimes, which distinguishes timelike, spacelike and null MTTs. We also investigate the stability of the \(2\)-spheres foliating the MTTs, which is determined by conditions on \(\sigma\).

\subsection{Relationship between \(C\) and the slope to the tangent to the MTT}

One of the very first investigations of the evolution of black holes, covariantly, utilizing the \(1+1+2\) formalism, was carried out by the authors in \cite{rit1}. This approach was generalized in \cite{as1}. Since in this paper we utilize the formalism in \cite{ib3}, but in context of the geometric and thermodynamic variables from the \(1+1+2\) formalism, it would make sense to show the relationship between the two approaches. While we do this (and provide examples) for LRS II spacetimes, this extends to the NNF spacetimes.

The approach in \cite{rit1} chooses a tangent vector \(\Psi^a=\alpha u^a + \beta e^a\), which lies on entirely on the MTT, defined as the curve \(\Theta_k=0\) in the \(\left[u,e\right]\) plane. Since the vector \(\nabla_a\Theta_k\) is normal to the MTT, then \(\Psi^a\nabla_a\Theta_k\) is zero, which gives the slope, \(\alpha/\beta\), of the MTT curve

\begin{eqnarray}\label{feb271}
\frac{\alpha}{\beta}=-\frac{\hat{\Theta}_k}{\dot{\Theta}_k}.
\end{eqnarray}
The sign of \(\alpha/\beta\) then determines if the MTT is future outgoing or future ingoing: if \(\alpha/\beta<0\), the MTT is said to be future ingoing, and if \(\alpha/\beta>0\) the MTT is said to be future outgoing. Now, from the definition of the function \(C\), lets look at for which values of \(\alpha/\beta\) we can causally characterize the MTT.

We can write the function \(C\) from \eref{03} as 

\begin{eqnarray*}
C=\frac{\dot{\Theta}_k + \hat{\Theta}_k}{\dot{\Theta}_k - \hat{\Theta}_k}.
\end{eqnarray*}
Dividing the numerator and denominator of  by \(\dot{\Theta}_k\) gives the function \(C\) in terms of the slope

\begin{eqnarray}\label{feb272}
C=\frac{1-\alpha/\beta}{1+\alpha/\beta}.
\end{eqnarray}

Multiplying \eref{feb272} by \(\left(1+\alpha/\beta\right)/\left(1+\alpha/\beta\right)\) (for \(\alpha/\beta\neq -1\)) gives 

\begin{eqnarray}\label{march520191}
C=\frac{1-\alpha^2/\beta^2}{\left(1+\alpha/\beta\right)^2}.
\end{eqnarray}
Then the MTT is spacelike, timelike of null if \(\alpha^2/\beta^2\) is less than \(1\), greater than \(1\) or equal to \(1\), respectively. For \(\alpha/\beta= -1\), we simply multiply \eref{feb272} by \(\left(1-\alpha/\beta\right)/\left(1-\alpha/\beta\right)\) and obtain

\begin{eqnarray}\label{march520192}
C=\frac{\left(1-\alpha/\beta\right)^2}{1-\alpha^2/\beta^2}.
\end{eqnarray}
The definitions of spacelike, timelike and null MTT follow as before.

\subsection{Characterization of MTT and the stability of MTS in some well known spacetimes}

We here consider the characterization of MTTs in the Robertson-Walker and Lemaitre-Tolman-Bondi models, as well as the stability of the MTS foliating these MTTs. 

\subsubsection{Robertson-Walker spacetimes}
\ \\
\ \\
As mentioned earlier, for certain solution types, the form of \eref{06}, when compared to the expression for \(C\), allows us to express certain quantities in terms of the geometric and thermodynamic variables. Consider the case of a timelike perfect fluid. The tangent to the timelike congruence is given by

\begin{eqnarray}\label{mi1}
u^a&=&Dk^a+\left(2D\right)^{-1}l^a.
\end{eqnarray}
The function \(C\) is expressed as 

\begin{eqnarray}\label{07}
C&=&\frac{1}{2D^2}\frac{\rho+p}{\left(\frac{1}{A}\right)+p-\rho},
\end{eqnarray}
for some function \(D\), \cite{ib3}, where \(A\) is the area of the \(2\)-spheres. The function \(C\) is thus determined by the sign of 

\begin{eqnarray}\label{08}
\frac{\rho+p}{\left(\frac{1}{A}\right)+p-\rho}.
\end{eqnarray}
We can directly compare \eref{07} to \eref{06} and write the area \(A\) as 

\begin{eqnarray}\label{09}
A=\frac{3}{\rho}D^2
\end{eqnarray}
(we note that for such solutions \(\Lambda,\mathcal{E}=0\)). Consider the Robertson-Walker spacetimes which assume an equation of state of the form \(p=\sigma\rho\). 
Putting \eref{09} in the denominator of \eref{08} we obtain

\begin{eqnarray}\label{mi2}
\left(\frac{1}{3D^2}+\left(\sigma-1\right)\right)\rho.
\end{eqnarray}

For a given solution (or class of solutions), knowing the denominator of \(C\) in \eref{06} determines \(D\). For LRS II spacetimes, \(D=1/\sqrt{2}\) (we can check \eref{mi1} against the expressions for \(k^a\) and \(l^a\) in \eref{oct20}).  This gives the area \(A\), of the marginally trapped \(2\)-spheres in the case of timelike perfect fluids as \(A=3/\left(2\rho\right)\). Inserting this in the condition in proposition \ref{rem1} and noting that both \(\mathcal{E}\) and \(\Lambda\) vanish, we obtain the results in \cite{ib3}, i.e. timelike and spacelike MTTs satisfy the cut-offs \(\rho-p>1/A\) and \(\rho-p\leq 1/A\) respectively.

For the Robertson-Walker spacetimes, we can explicitly write the quantity \(C\) in terms of \(\sigma\),
\begin{eqnarray}\label{mi3}
C=\frac{3\left(\sigma+1\right)}{\left(3\sigma-1\right)}.
\end{eqnarray}
Of course from \eref{mi3}, any real value of \(\sigma\) is suppose to fix \(C\) everywhere on the MTT. However, in \cite{ib3} the author showed that in certain cases, for higher values of \(\sigma\), the value of \(C\) may vary on the MTT, in particular for \(\sigma=2\). But the horizon remains spacelike. We see that the causal characterization of an MTT corresponds to the following bounds on \(\sigma\): For a timelike MTT, \(-1<\sigma<\frac{1}{3}\), a spacelike MTT, \(\sigma>\frac{1}{3}\) or \(\sigma<-1\), and for null MTT, \(\sigma=-1\) or \(\sigma=\frac{1}{3}\). The approach used here relatively easily obtains these results, which agree with those of Ben-Dov and Senovilla \cite{bend1,sen3}.

The formula in \eref{mi3} could have been directly obtained from \eref{06}: For timelike fluids \(C\) becomes
\begin{eqnarray}\label{mi4}
C&=&\frac{3\left(\rho+p\right)}{\left(3p-\rho\right)}.
\end{eqnarray}
With an equation of state of the form \(p=\sigma\rho\), \eref{mi4} is reduced to \eref{mi3}. An interest of this exercise is the form that \(A\) takes. Consider stability of the marginally trapped \(2\)-spheres of the Robertson-Walker spacetimes. By stability we mean that if a marginally trapped \(2\)-surface \(S\) is deformed outward, the associated outgoing null expansion scalar is non-negative and somewhere positive on the marginally trapped \(2\)-sphere (see references \cite{and1,and2,yau1,jang1,gg10,gg11}). Given a function \(\varphi\) on \(S\), stability of \(S\) is determined by the sign of the principal eigenvalue of the stability operator \(L\) acting on \(\varphi\). For spherically symmetric spacetimes, this eigenvalue is given by 

\begin{eqnarray}\label{dddppp1}
\lambda=8\pi\left(\frac{1}{2\mathcal{A}}-T_{ab}k^al^b\right),
\end{eqnarray}
(see \cite{ib3}). For timelike fluids, we can then write \(\lambda\) as

\begin{eqnarray}\label{mi5}
\lambda=-8\pi\left(\frac{2}{3}\rho+p\right).
\end{eqnarray}
Since energy density is positive, for spacetimes in the class of timelike perfect fluids, the MTS are stable if \(\rho\leq -\frac{3}{2}p\) and strictly stable if \(\rho< -\frac{3}{2}p\). We see that stability of the marginally trapped \(2\)-spheres is only obtained under the condition of negative pressure. For the Robertson-Walker spacetimes, the condition for stability reduces to

\begin{eqnarray}\label{mi6}
\sigma\leq-\frac{2}{3}.
\end{eqnarray}
From \eref{mi6} it is evident that the marginally trapped \(2\)-spheres of null and contracting, spacelike and expanding MTTs, as well as the \(2\)-spheres of  MTTs in the dust filled universe (this was also recently considered in \cite{as1} in context of the \(1+1+2\) semitetrad formalism) are all unstable. 

Our approach has allowed us to completely characterize the marginally trapped \(2\)-spheres of MTTs in these cosmological models in terms of stability, again, relatively easily. It is also possible to identify the SFOTHS in the Robertson-Walker spacetimes. We have that \eref{10} reduces to the condition 
\begin{eqnarray}\label{ams3}
-\sigma<1<3\sigma.
\end{eqnarray}
It is clear from \eref{ams3} that the SFOTHs in these models are precisely those spacelike MTTs with \(\sigma>\frac{1}{3}\), containing unstable marginally trapped \(2\)-spheres. 

\subsubsection{Lemaitre-Tolman-Bondi dust model}
\ \\
\ \\
By the same token, let us consider the the Lemaitre-Tolman-Bondi (LTB) dust model \cite{rit50,rit51,rit52}, a gravitational collapse model violating the cosmic censorship conjecture, with interior metric given by the line element
\begin{eqnarray}\label{beepee1}
ds^2=-dt^2 + \frac{R'^2}{1 - r^2b_0}dr^2 + R^2\left(d\theta^2 + \sin^2\theta d\phi^2\right).
\end{eqnarray}
Here \(R=R\left(t,r\right)\) is the area radius of the collapsing shell and \(b_0=b_0\left(r\right)\) is their energy profile. The only non-zero matter and thermodynamic quantities are \(\rho\) and \(\mathcal{E}\). Evolution of MTTs in the LTB model has been considered in \cite{ib3}. We also consider the evolution of MTTs in this model, but in terms of the geometric and thermodynamic variable as we have done throughout this paper. The function \(C\) for the LTB model is given by 
\begin{eqnarray}\label{jan4191}
C=-\frac{\rho}{\frac{1}{3}\rho + 2\mathcal{E}},
\end{eqnarray}
where \(\rho\) is given by 
\begin{eqnarray}\label{thucidy1}
\rho=\frac{\left(r^3\mathcal{M}\right)'}{R^2R'},
\end{eqnarray}
with \(\mathcal{M}=\mathcal{M}\left(r\right)\) being the Misner-Sharp mass \cite{rit50,rit51,rit52}. Thus the condition for an MTT to be spacelike, timelike, or null is given by
\begin{eqnarray*}
\rho<-6\mathcal{E},\quad \rho>-6\mathcal{E},\quad or \quad \rho=-6\mathcal{E},
\end{eqnarray*}
respectively. 

The scalar \(\lambda\), whose sign determines the stability of the MTS, is given by 

\begin{eqnarray}\label{jan4192}
\lambda=-4\pi\left(\frac{1}{3}\rho + 2\mathcal{E}\right).
\end{eqnarray}
From \eref{dddppp1}, stability of the MTS thus requires that

\begin{eqnarray}\label{jan4193}
\rho<-6\mathcal{E},
\end{eqnarray}
on the MTS. Of course then spacelike MTTs are necessarily foliated by stable MTS, and the MTS foliating timelike MTTs are unstable. The null MTTs are marginally stable. 

Let us now consider the relationship between stability and shell crossing. For the LTB model, \(\mathcal{E}\) is given by 

\begin{eqnarray}\label{jan4194}
\mathcal{E}=\frac{1}{3}\rho - \frac{r^3\mathcal{M}}{R^3},
\end{eqnarray}
where \(R\) is given as \(R=r\left(1-\sqrt{\mathcal{M}}t\right)^{2/3}\) \cite{rit50,rit51,rit52}. The condition for stability can be expressed as a condition on \(\mathcal{M}'\): Upon inserting \eref{jan4194} into \eref{jan4193}, we obtain the condition for
 stability (which is also an equivalent condition for the MTT to be spacelike) as
 
\begin{eqnarray}\label{jan4195}
\rho < \frac{2\mathcal{M}}{\left(1-\sqrt{\mathcal{M}}t\right)^2},
\end{eqnarray}
and from \eref{thucidy1}, \eref{jan4195} can be rewritten as 

\begin{eqnarray}\label{jan4196}
\mathcal{M}' <&\frac{3r\mathcal{M}\sqrt{\mathcal{M}}\left(1-\sqrt{\mathcal{M}}t\right)}{3r^2\left(1-\sqrt{\mathcal{M}}t\right) + 2t\mathcal{M}},
\end{eqnarray}
as we move away from the central singularity, (\(t_c<t<\infty\)), where \(t_c=1/\sqrt{\mathcal{M}}\) is the time for collapse (comoving time) to zero area (\(R=0\)). The condition for no shell crossing is given by

\begin{eqnarray*}
t'_c=-\frac{\mathcal{M}'}{2\mathcal{M}^{\frac{3}{2}}}\geq 0,
\end{eqnarray*}
which requires \(\mathcal{M}'\leq 0\). It is clear that stability of the leaves of the MTT ensures no shell crossing. 

Given an MTT \(H\), how does the causal character of \(H\) affect flux quantities associated with \(H\)? In the next section we consider surface gravity calculated on the MTT for various classes of spacetimes.

\section{Surface gravity of MTTs}\label{0n}

Surface gravity is a fundamental notion in the formulation of the laws of black hole mechanics. As such, understanding their evolution as a black hole evolves has been of keen interest \cite{ash7,ash8,shh1,by1}. In this section, using the formulation in \cite{hay3,ash7,ash8,ash9} (see also \cite{ib1, ib2}), coupled with adapting the calculations to quantities from the \(1+1+2\) splitting, we calculate the surface gravity for black hole horizons in the LRS II and NNF spacetimes, as well as a general \(4\)-dimensional spacetime. Under this formulation we prove a standard result in black hole mechanics, specifically the third law of black hole thermodynamics, for the LRS II class.

One defines the connection on the normal bundle to the \(2\) -surface \(S\) as 

\begin{eqnarray}\label{nov1510}
\omega_a:=-\ l_b\nabla_ak^b,
\end{eqnarray} 
For the case of an IH, the surface gravity is given by the contraction of \(\omega_a\) by the outgoing null normal:

\begin{eqnarray}\label{21}
\kappa&=&\omega_ak^a.
\end{eqnarray} 
Let us first consider the case of the LRS II class of spacetimes. Computing \(\omega_a\) we obtain \(\omega_a=-\left(Au_a-\left(\frac{1}{3}\Theta+\Sigma\right)e_a\right)\), which gives

\begin{eqnarray}\label{22}
\kappa_{LRS\ II}&=&\frac{1}{\sqrt{2}}\left(A+\frac{1}{3}\Theta+\Sigma\right).
\end{eqnarray} 
In the time symmetric case, \(\left(\frac{1}{3}\Theta+\Sigma\right)\) is zero and the surface gravity proportional to the accelaration \(A\), which is the expected result.

For the NNF, \(\omega_a\) is given by

\begin{eqnarray}\label{sg1}
\fl\eqalign{\omega_a=&\left[-A+\frac{\sqrt{2}}{2}\left(A^bl_b+\alpha^bl_b\right)\right]u_a+\left[\frac{1}{3}\Theta+\Sigma+\frac{\sqrt{2}}{2}\left(\Sigma^bl_b+\varepsilon^b_m\Omega^ml_b+a^bl_b\right)\right]e_a\\
&+\left[\Sigma_a-\varepsilon_{am}\Omega^m-\frac{\sqrt{2}}{2}\left(\left(\Omega+\xi\right)\varepsilon^b_a+\Sigma^b_a+\zeta^b_a\right)l_b\right]\\
&=-Au_a+\left(\frac{1}{3}\Theta+\Sigma\right)e_a+\Sigma_a-\varepsilon_{am}\Omega^m.}
\end{eqnarray}
The surface gravity on an IH in the NNF class is then given by 

\begin{eqnarray}\label{sg2}
\kappa_{NNF}&=&\frac{1}{\sqrt{2}}\left(A+\frac{1}{3}\Theta+\Sigma\right)=\kappa_{LRS\ II}.
\end{eqnarray}

In general, the surface gravity on a leaf of the foliation, applicable to different horizon types, is given by

\begin{eqnarray}\label{23}
\kappa^{\nu}=-l_bX^a\nabla_aX^b,
\end{eqnarray} 
where \(\nu\) is the foliation label and \(X^a\) is the vector given in \eref{01}, which can be expanded as

\begin{eqnarray}\label{24}
\kappa^{\nu}=k^a\omega_a+C\left(X^al_b\nabla_al^b - l^a\omega_a\right),
\end{eqnarray} 

We see that on an isolated horizon (\(C=0\)), \eref{24} becomes \eref{21}. In fact for a slowly evolving horizon (see \cite{ib1} for discussions on slowly evolving horizons), \eref{24} gives a splitting of the surface gravity with \(C\) being - taken as the evolution rate - sufficiently less that \(1\), and \(k^a\omega_a\) is constant. Then a slowly evolving horizon will be interpreted simply as one satisfying the conditions in \cite{ib1} where \(\epsilon\) is replaced by \(C\), and so are simply those horizons for which \(C\) satisfies \(0<|C|\ll1\) along with the additional conditions.

From \eref{24} it is now obvious that in general the surface gravity will not be constant  with \(C\left(X^al_b\nabla_al^b-l^a\omega_a\right)\) being the evolving component of \(\kappa^{\nu}\). 

For the LRS II class of spacetimes, \(l_b\nabla_al^b=0\) and thus, on a leaf, the surface gravity can be expressed as
\begin{eqnarray}\label{0b}
\kappa^{\nu}_{LRS\ II}=\frac{1}{\sqrt{2}}\left[\left(1-C_{LRS\ II}\right)A+\left(1+C_{LRS\ II}\right)\left(\frac{1}{3}\Theta+\Sigma\right)\right].
\end{eqnarray}
In the case of an NEH with \(C_{LRS\ II}=0\), \eref{0b} reduces to \eref{22}. 

As an example of the application of \eref{0b} to horizons that are not NEHs, we again consider the OS collapse case. In this case \(A,\Sigma\) are both vanishing and we have 

\begin{eqnarray*}
\kappa^{\nu}\propto -\frac{2}{3}\Theta.
\end{eqnarray*}
Of course since it is known that \(\kappa> 0\) we therefore expect \(\Theta<0\) and so the horizon in this case is necessarily collapsing, and thus timelike. 

For the NNF class, we have 

\begin{eqnarray*}
\fl\eqalign{X^al_b\nabla_al^b=&\frac{\left(1+C_{NNF}\right)}{2}\left(\Sigma^b+\varepsilon^b_m\Omega^m+\left(\Omega-\xi\right)e^a\varepsilon^b_a+e^a\Sigma^b_a-a^b-e^a\zeta^b_a\right)l_b\\
&+\frac{\left(1-C_{NNF}\right)}{2}\left(A^b-\alpha^b+u^a\Sigma^b_a-u^a\zeta^b_a+\left(\Omega-\xi\right)u^a\varepsilon^b_a\right)l_b\\
&=0,}
\end{eqnarray*}
and

\begin{eqnarray*}
l^a\omega_a=\frac{1}{\sqrt{2}}\left[A-\left(\frac{1}{3}\Theta+\Sigma\right)\right].
\end{eqnarray*}
We therefore have \eref{24} reducing to

\begin{eqnarray}\label{nov181}
\eqalign{\kappa^{\nu}_{NNF}&=\frac{1}{\sqrt{2}}\left(A+\frac{1}{3}\Theta+\Sigma\right)+\frac{1}{\sqrt{2}}C_{NNF}\left(-A+\frac{1}{3}\Theta+\Sigma\right)\\
&=\frac{1}{\sqrt{2}}\left[\left(1-C_{NNF}\right)A+\left(1+C_{NNF}\right)\left(\frac{1}{3}\Theta+\Sigma\right)\right].}
\end{eqnarray}

Let us consider under what condition the horizon of a black hole spacetime in the LRS II class (and by extension NNF) in general will have a vanishing \(\kappa^{\nu}\) on the horizon. Suppose \(\kappa^{\nu}_{LRS\ II}=0\) on \(H\). Then setting \eref{nov181} to zero gives

\begin{eqnarray}\label{nov182}
C_{LRS\ II}=\frac{A+\left(\frac{1}{3}\Theta+\Sigma\right)}{A-\left(\frac{1}{3}\Theta+\Sigma\right)}. 
\end{eqnarray}
 Equating \eref{nov182} to \eref{06} we obtain the following constraint equation
 
\begin{eqnarray}\label{nov183}
\left[A-\left(\frac{1}{3}\Theta+\Sigma\right)\right]\left[-\frac{4}{3}\rho-\Pi+2\left(Q-\mathcal{E}\right)\right]=0,
\end{eqnarray}
so either \(A=\left(\frac{1}{3}\Theta+\Sigma\right)\) or \(-\frac{4}{3}\rho-\Pi+2\left(Q-\mathcal{E}\right)=0\). If \(A=\left(\frac{1}{3}\Theta+\Sigma\right)\), then \(C_{LRS\ II}=\infty\) from \eref{nov182}. But from the definition of \(\kappa^{\nu}_{LRS\ II}\), \(A=\left(\frac{1}{3}\Theta+\Sigma\right)\) implies \(\left(\frac{1}{3}\Theta+\Sigma\right)=0\) (since \(\kappa^{\nu}_{LRS\ II}\) is 0). This then further implies \(A=0\) which makes \(C_{LRS\ II}\) indeterminate, and we have a contradiction. We therefore assume \(A\neq\left(\frac{1}{3}\Theta+\Sigma\right)\) and consider the second case, \(-\frac{4}{3}\rho-\Pi+2\left(Q-\mathcal{E}\right)=0\). This condition implies \(C_{LRS\ II}=1\), which would then imply that 

\begin{eqnarray}\label{nov185}
\left(\frac{1}{3}\Theta+\Sigma\right)=0,
\end{eqnarray}
(from either equation \eref{nov181} or \eref{nov182}). Since \(A\neq\left(\frac{1}{3}\Theta+\Sigma\right)\) and \(\left(\frac{1}{3}\Theta+\Sigma\right)=0\), \(A\neq 0\). We therefore have that \(\kappa^{\nu}=0\) if \(A\neq 0\) and \(\Theta=-\phi\) on the horizon, recalling that on the horizon,

\begin{eqnarray}\label{0000000001} 
\Theta_k=\frac{2}{3}\Theta-\Sigma+\phi=0.
\end{eqnarray}
Combining \eref{nov185} and \eref{0000000001} gives \(\phi=-\Theta\) and \(\Sigma=\frac{1}{3}\phi\). But \(C_{LRS\ II}=1\) implies \(\Theta>0\) (expanding). As the \(2\)-surfaces foliating the horizon are marginally trapped, the expansion scalar in the ingoing null direction is negative, i.e. \(\Theta_l=\frac{2}{3}\Theta-\Sigma-\phi<0\). However, we have \(\Theta_l=2\Theta\). So for \(\Theta>0\), we have that \(\Theta_l>0\), which is a contradiction. We conclude that

\begin{proposition}
\textit{We cannot have a black hole spacetime in the LRS II class (and by extension the NNF class) with \(\kappa^{\nu}\) vanishing on the horizon}, 
\end{proposition}
which is just the third law of black hole thermodynamics \cite{shh1}. Though restricted to the LRS II and NNF classes, this provides an explicit proof of the third law of black hole thermodynamics \cite{shh2,shh3,shh1}.

We now consider these quantities for a general \(4\)-dimensional spacetime. The outgoing null normal \(k^a\) in this case is as defined as in \eref{nov151}. Using \eref{nov1510}, \(\omega_a\) is calculated as

\begin{eqnarray}\label{}
\omega_a=-Au_a+\left(\frac{1}{3}\Theta+\Sigma\right)e_a-\tau_a\Theta_k\vline_{LRS\ II}+l_b\nabla_a\tau^b+\Sigma_a-\varepsilon_{am}\Omega^m,
\end{eqnarray}
which gives the surface gravity as 

\begin{eqnarray}\label{nov1512}
\kappa=\frac{1}{\sqrt{2}}\left[\left(A+\frac{1}{3}\Theta+\Sigma\right)-l_b\left(\dot{\tau}^b+\hat{\tau}^b+\breve{\tau}^b\right)\right],
\end{eqnarray}
where the notation \(\breve{\ast}\) is introduced to denote the derivative along the direction \(\tau^a\). So \( \breve{\tau}^b=\tau^a\nabla_a\tau^b\). From \eref{nov1512} we see that the surface gravity on a NEH in a general \(4\)-dimensional spacetime coincides with the LRS II class if \(l_b\breve{\tau}^b=-l_b\left(\dot{\tau}^b+\hat{\tau}^b\right)\). The \(1+1+2\) decomposition of the surface gravity on a leaf of a general horizon type is calculated as

\begin{eqnarray}\label{nov191}
\eqalign{\kappa^{\nu}=&\frac{1}{\sqrt{2}}\left[\left(A+\frac{1}{3}\Theta+\Sigma\right)-l_b\left(\dot{\tau}^b+\hat{\tau}^b+\breve{\tau}^b\right)\right]\\
&+\frac{1}{\sqrt{2}}C\left[-A+\left(\frac{1}{3}\Theta+\Sigma\right)+\frac{1}{2}\left(C-3\right)\dot{\tau}^bl_b\right]\\
&+\frac{1}{\sqrt{2}}C\left[\frac{1}{2}\left(C+3\right)\left(\hat{\tau}^b+\breve{\tau}^b\right)l_b+\frac{1}{2}\tau_b\left(V_2^b-V_1^b\right)\right]\\
&+\frac{1}{\sqrt{2}}C\left[\frac{1}{2}\left(1+C\right)V+\tau_b\left(\Sigma^b-\varepsilon_a^b\Omega^a\right)\right]\\
=&\frac{1}{\sqrt{2}}\left[\left(1-C\right)A+\left(1+C\right)\left(\frac{1}{3}\Theta+\Sigma+\frac{1}{2}CV\right)\right]\\
&+\frac{1}{2\sqrt{2}}\left[\left(C^2-3C-1\right)\dot{\tau}^b-\left(C^2+3C-1\right)\left(\hat{\tau}^b+\breve{\tau}^b\right)\right]l_b\\
&+\frac{1}{2\sqrt{2}}\left[C\left(V_2^b-V_1^b\right)+\left(\Sigma^b-\varepsilon^b_a\Omega^a\right)\right]\tau_b,}
\end{eqnarray}
where

\begin{eqnarray*}
V_1^b=A^b-\alpha^b, \ V_2^b=a^b-\Sigma^b-\varepsilon^b_c\Omega^c, \ V=\frac{1}{2}\tau^2\Theta_k\vline_{LRS\ II}.
\end{eqnarray*}
Again when \(C=0\) \eref{nov191} reduces to the form of \eref{nov1512}.

In fact the expression for the surface gravity of IH in the LRS II and NNF classes  has consequence for the definition of the temperature of the gravitational field as formulated in the gravitational entropy proposal paper by Clifton \textit{et. al} \cite{tgr1}. Let us define a new one-form \(\omega'_a\) as

\begin{eqnarray}
\omega'_a&=&-\frac{1}{2\pi\sqrt{2}}\left(l_b\nabla_a k^b+l_b\nabla_a l^b\right)\nonumber\\
&=&\frac{1}{\pi\sqrt{2}}\left(\omega_a+l_b\nabla_a e^b\right),
\end{eqnarray}
where \(\omega_a\) is as defined in \eref{nov1510}. We can then write the temperature of the gravitational field as

\begin{eqnarray}\label{nov271}
T_{grav}&=&k^a\omega'_a\nonumber\\
&=&\frac{1}{\pi\sqrt{2}}\left(k^a\omega_a+l_bk^a\nabla_ae^b\right)\nonumber\\
&=&\left(\frac{2-\sqrt{2}}{4\pi}\right)\left(A+\frac{1}{3}\Theta+\Sigma\right),
\end{eqnarray}
(also see \cite{rit10} for discussion of the temperature of the gravitational field in context of the \(1+1+2\) splitting). The first term in the parenthesis of the second line of \eref{nov271} is clearly associated with the horizon and as such, in an LRS II spacetime admitting an IH, the temperature will have contribution from the horizon. 

We mention that in general one has to exercise caution when assigning physical meaning to all these quantities precisely because of the fact that the true values are tied to the scaling of the vector \(k^a\).

Also we can give the following geometrical interpretation to \(T_{grav}\) (up to scale): \(T_{grav}\) is the sum of the accelerations of the null normal vector fields \(k^a\) and \(l^a\) along the null direction \(\l^a\).

\section{Conclusion}\label{err}

In this work we have considered various aspects of the dynamics of black hole horizons. Work by Booth established that the matter content of the horizon determines the causal character of the horizon (at least for the spherically symmetric case he considered). This work treats this same problem using the formulation by Booth and coauthors in \cite{ib3}, but in context of the geometric and thermodynamic variables from the \(1+1+2\) splitting. The form that the function \(C\) takes - whose sign determines the causal character of the horizon -  indeed establishes that the matter content determines the causal character of the horizon. Stated another way, the causal character of the horizon is determined by the relationship between the matter and thermodynamic variables, even for more general cases. The form of \(C\) in our treatment further allows us to recover established results as well as provide us with new insights into the properties of the horizon in a straightforward and relatively simple way. This is demonstrated by the recovery of the bounds of the equation of state parameter \(\sigma\), for which horizons in the Robertson-Walker spacetimes are timelike, spacelike of non-expanding. We have also been able to determine the values of \(\sigma\) for which MTS in the Robertson-Walker spacetimes are stable as well as classify the SFOTHs. In particular, it is seen that the existence of stable MTS are possible only if the isotropic pressure is negative. We then went on to show that for the LTB model, a relationship between the energy density and electric part of the Weyl curvature \(\mathcal{E}\), gives the causal classification of the MTTs. It was further shown that the MTS foliating the spacelike MTTs are necessarily stable, and that this stability guarantees no shell crossing.

The form of the surface gravity obtained in our treatment has provided use a useful mean to verify the third law of black hole thermodynamics for a restricted class of spacetimes. We give an explicit proof of the law for LRS II spacetimes (the proof can also be carried over to NNF spacetimes). For an isolated horizon, we have an expression that allows us to give a geometric interpretation of the temperature of a gravitational field.

\ack
We would like to thank the anonymous referee(s) for the very helpful comments and suggestions. AS and RG are supported by National Research Foundation (NRF), South Africa. SDM 
acknowledges that this work is based on research supported by the South African Research Chair Initiative of the Department of Science and Technology and the National Research Foundation.

\end{document}